\documentclass[10pt, journal]{IEEEtran}
\usepackage{tikz}
\usepackage{acronym}
\usetikzlibrary{shapes,snakes}
\usepackage{empheq}
\usepackage{graphicx}
\usepackage{epstopdf}
\usepackage{multirow}
\usepackage{tabu}
\usepackage{array}
\usepackage{url}
\usepackage{amsmath}
\usepackage{cleveref}
\usepackage{amssymb}
\usepackage{amsthm}
\usepackage{epstopdf}
\usepackage{amsfonts}
\usepackage{bm}
\usepackage{algorithmic}
\usepackage[ruled,vlined,linesnumbered]{algorithm2e}
\usepackage{graphicx}
\usetikzlibrary{arrows}
\usepackage{cite}
\usepackage{caption}
\usepackage{subcaption}
\usepackage{amssymb}
\usepackage{tabulary}
\usepackage{booktabs}

\RequirePackage{xcolor}
\def\current@color{Black}

\usepackage[justification=centering]{caption} 
\DeclareCaptionLabelFormat{algnonumber}{algorithm}
\tikzstyle{int}=[draw, fill=white!20, minimum size=2em]
\tikzstyle{init} = [pin edge={to-,thin,black}]

\newtheorem{thm}{Theorem}

\newtheorem{lemma}{Lemma}

\newcommand\blfootnote[1]{%
	\begingroup
	\renewcommand\thefootnote{}\footnote{#1}%
	\addtocounter{footnote}{-1}%
	\endgroup
}

\DeclareRobustCommand*{\IEEEauthorrefmark}[1]{%
	\raisebox{0pt}[0pt][0pt]{\textsuperscript{\footnotesize\ensuremath{#1}}}}

\newcommand{\algrule}[1][.2pt]{\par\vskip.5\baselineskip\hrule height #1\par\vskip.5\baselineskip}

\makeatletter
\def\munderbar#1{\underline{\sbox\tw@{$#1$}\dp\tw@\z@\box\tw@}}
\makeatother

\acrodef{AoI}{Age of Information}
\acrodef{IoT}{Internet of Things}
\acrodef{DPP}{Drift-Plus-Penalty}
\acrodef{DTMC}{Discrete-Time Markov Chain}
\acrodef{w.p.}{with probability}
\acrodef{EVD}{EigenValue-Decomposition}
\acrodef{MDP}{Markov Decision Process}
\acrodef{CSI}{Channel State Information}
\acrodef{CMDP}{Constrained Markov Decision Process}
\acrodef{i.i.d}{independent and identically distributed}
\acrodef{ACK}{acknowledgement}
\acrodef{RVIA}{Relative Value Iteration Algorithm}
\acrodef{VI}{Value Iteration}
\acrodef{w.r.t}{with respect to}

\IEEEoverridecommandlockouts                          

\begin{document}
\title{
  Goal-oriented Policies for Cost of Actuation Error Minimization in Wireless Autonomous Systems
}

\author{
   \IEEEauthorblockN{Emmanouil Fountoulakis, Nikolaos Pappas, and Marios Kountouris}
}
\maketitle

\immediate\write18{echo $PATH > tmp1}
\immediate\write18{/Library/TeX/texbin/epstopdf > tmp2} 

\maketitle
\thispagestyle{plain} 
\pagestyle{plain}
\begin{abstract}
    We consider the minimization of the cost of actuation error under resource constraints for real-time tracking in wireless autonomous systems. A transmitter monitors the state of a discrete random process and sends updates to a receiver over an unreliable wireless channel. The receiver then takes actions according to the estimated state of the source. For each discrepancy between the real state of the source and the estimated one, we consider a different cost of actuation error. This models the case where some states, and consequently the corresponding actions to be taken, are more important than others. We provide two algorithms, a first one reaching an optimal solution but of high complexity, and a second low-complexity one that provides a suboptimal solution. Our simulation results evince that the performance of the two algorithms are quite close.
\end{abstract}

\section{Introduction}
\blfootnote{
This work has been performed while E. Fountoulakis was with the Communication Systems Department, EURECOM, France. He is now with Ericsson, Sweden. Email: emmanouil.fountoulakis@ericsson.com. 
N. Pappas is with the Department of Computer and Information Science, Link\"oping University, Sweden. Email: nikolaos.pappas@liu.se. 
M. Kountouris is with the Communication Systems Department, EURECOM, France. Email: marios.kountouris@eurecom.fr. 

The work of E. Fountoulakis and M. Kountouris has received funding from the European Research Council (ERC) under the European Union’s Horizon 2020 research and innovation programme (Grant agreement No. 101003431). The work of N. Pappas has been supported in part by the Swedish Research Council (VR), ELLIIT, Zenith, and the European Union (ETHER, 101096526). 
}

Emerging cyber-physical and real-time autonomous systems are envisioned to introduce various applications and services, in which information distilled from measurements or observations is valuable when it is fresh, accurate, and useful to the specific goal of the data exchange. In this context, a relevant yet challenging problem is that of remote real-time tracking and actuation driven by sampled and potentially delayed measurements transmitted over a wireless channel using limited resources.
    
Conventional communication system design has mainly remained agnostic to the \textit{significance} of transmitted messages, in particular at the physical and medium access layers. The optimization of system performance has been dominated by metrics such as throughput, delay, and packet drop rate. Although these performance metrics have turned out to be instrumental for enabling reliable and efficient communication, they fall short of differentiating the packets according to their information content and its value. 
A recently developed metric, named \ac{AoI}, has been proposed to measure the freshness and the timeliness of information \cite{kaul2012real,kosta2017age, yates2021age}. However, baseline \ac{AoI}-based metrics do not take into account the source evolution and the significance of the generated information with respect to the communication task/goal and the context. Several variants of \ac{AoI} have been proposed for tackling the problem of remote estimation in status update systems \cite{huang2020real, sun2019sampling, tang2022sampling, zheng2020urgency}.
Nevertheless, the aforementioned works do not consider the \textit{cost of actuation error}, as they mainly focus on the discrepancy between the source and the estimated value of the process at the destination. A recently proposed approach, which is also adopted in this paper, takes into account the semantics of the information, i.e., significance, goal-oriented usefulness, and contextual importance of information as a means to leverage the synergy between data generation and processing, information transmission, and signal reconstruction \cite{kountouris2021semantics,DenizJSAC,stavrou2022rate,Kaibin21Sem}.

In this work, we consider the problem of real-time tracking and estimation of an information source from a remote actuator. A transmitter samples and sends information about the state of a source in the form of status update packets to a remote actuator (receiver) over an unreliable wireless channel. The actuator takes actions depending on the estimated state of the remote source. We also consider that the transmitter has limited resources, which prevents it from sampling and transmitting updates continuously. 
This paper extends the results of \cite{pappas2021goal, mehrdad2023}, where the problem of remote monitoring of a discrete Markov source is considered and semantics-empowered policies are proposed to significantly reduce both the real-time reconstruction and the cost of actuation errors, as well as the amount of ineffective updates.
Specifically, we consider a more general discrete stochastic source process and resource constraints, which make the solution essentially different. 
The problem is formulated as a \ac{CMDP}, and two \textit{goal-oriented semantic-aware} policies are proposed. A key takeaway is that it is optimal for the transmitter to remain silent even if there is a discrepancy between the actual state of the source and its estimate at the receiver, due to the delay induced by the wireless channel, the high transition probability, and the large actuation error. 

\section{System Model}
We consider a time slotted communication system in which a transmitter monitors a discrete random process and sends status updates to a receiver over an error-prone wireless channel. Let $t \in \mathbb{Z}_{> 0}$ denote the $t$-th slot. The receiver operates as a remote actuator and performs actions according to the estimated state of the source. The state of the process is modeled by a \ac{DTMC} $\left\{X_t\right\}_{t\in Z_{>0}}$, and is assumed to be ergodic. The state of the source takes values from the set $\left\{0,1,\ldots, N\right\}$, where $N \in \mathbb{Z}_{>0}$. Each state corresponds to a specific action that has to be performed by the actuator. 

The channel realization is denoted by $h_t$, and is equal to $1$ if a packet is successfully received at time slot $t$ and $0$ otherwise. The success probability is defined as $p_s = \Pr(h_t=1)$, and the failure probability as $p_f = \Pr(h_t=0) = 1 - p_s$. For every successful transmission, the receiver updates its information regarding the state of the source with a new estimate denoted by $\hat{X}_t$. The receiver sends an \ac{ACK}/negative \ac{ACK} for successful/failed transmissions. We assume that ACK/NACK information is sent and received instantaneously and error free. If the receiver does not successfully receive an update, it uses its previous estimate as the current one, i.e., $\hat{X}_{t+1} = \hat{X}_t$. We consider that the sampling and transmission processes take a time slot to be performed. Therefore, the receiver receives an update from the transmitter with one slot delay, if a transmission is successful at time slot $t$, and the actuator updates its state at slot $t+1$.


The transmitter generates a status update $X_t$ by sampling the source \textit{at will}. The decision to sample and transmit at time slot $t$ is denoted by $\alpha_t$, where
\begin{align}\label{decvar: alpha}
		\alpha_t =
	\begin{cases}
	 1\text{, if the source is sampled and its state transmitted,}\\
						0\text{, otherwise.}
	\end{cases}
\end{align}
\subsection{Performance metrics}
We consider that the actuator (receiver) takes actions according to the estimated state of the source. If the estimated state is different from the real state of the source, an actuation error occurs depending on a pre-defined function. The cost of actuation error captures the \textit{significance} (semantics) of the error at the point of actuation. Note that some errors may have a larger or a more critical impact than others. Let $C_{i,j}$ denote the cost of being in state $i$ at the source, and in state $j$, estimated at the receiver, at time slot $t$, i.e., $X_t = i$ and $\hat{X}_t = j$. We assume that the costs $C_{i,j}$ are given and remain the same over the time horizon.
Furthermore, for every sampling and transmission actions, we consider a cost $c$. This cost can represent, for instance, the power consumption for both sampling and transmission procedures.

\section{Problem Formulation}
The objective of this work is to minimize the average total cost of actuation error under average resource constraints. The expected time averages of the transmission and actuation costs are defined as
\begin{align}
	\bar{c} \triangleq \lim\limits_{T\rightarrow \infty} \frac{1}{T} \sum\limits_{t=1}^T \mathbb{E}\left\{ \alpha_t c\right\}\text{, } 
	\bar{C} \triangleq \lim\limits_{T\rightarrow \infty } \frac{1}{T} \sum\limits_{t=1}^T \mathbb{E} \left\{C_{i,j}^t\right\}\text{,}
\end{align}
respectively.
To this end, we formulate our stochastic optimization problem as
	\begin{align}\label{optproblem}
		\min\limits_{\pi} \quad \bar{C}^\pi 	\text{, s.~t.,} \quad   \bar{c}^\pi \leq c_\text{max}\text{,}
	\end{align}
where $\pi$ is the policy that decides the rule of selecting the right value $\alpha_t$ at every time slot $t$, and $c > 0$ is the time-averaged cost constraint. The problem in \eqref{optproblem} is a \ac{CMDP}. The system state is described by tuple $S_t = (X_t,\hat{X}_t, C_{t})$, actions $\alpha_t \in \mathcal{A}$, where $\mathcal{A} = \left\{0,1\right\}$, and the transition matrix is described by $P_{i,j} = \Pr\left\{X_{t+1}=j|X_{t}=i\right\}$. 
We assume that the transmitter has knowledge of the channel and source statistics.

\section{Proposed Algorithms}
In this section, we provide two optimization algorithms for solving problem \eqref{optproblem} optimally and suboptimally. 
\subsection{Optimal Solution}
The problem in \eqref{optproblem} is a \ac{CMDP}, which is, in general, difficult to be solved \cite{altman1999constrained}. In order to solve the constrained problem, we relax the constraints in \eqref{optproblem} by utilizing Lagrangian multipliers. We show that this approach can provide the optimal solution. 

We define the \textit{Lagrangian} function as 
\begin{align}\label{eq:LagrangianCostTrError}
	\mathcal{L}(\pi,\lambda) = \lim\limits_{T\rightarrow \infty}\frac{1}{T} \sum\limits_{t=1}^{T} \mathbb{E}_\pi \left\{C_{t} + \lambda \alpha_t c \right\} - \lambda c\text{,}
\end{align}
where the immediate cost is
$
	f(S_t) = C_{t} + \lambda\alpha_t c\text{.}
$
In order to proceed with the solution in \eqref{optproblem}, we consider the following optimization problem
\begin{align}\label{opt:dual}
	\min\limits_{\pi \in \Pi} \mathcal{L}(\pi,\lambda)\text{,}
\end{align}
for any given $\lambda\geq 0$\footnote{For $\lambda=0$, one may expect that the optimal policy is to always transmit because the sampling and the transmission processes are costless. However, our simulation results show that this is not always optimal even for cost-free transmissions.}. Since $\lambda c$ is independent of the chosen policy $\pi$, the problem in \eqref{opt:dual} is equivalent to the following optimization problem
\begin{align}\label{opt:dualWithoutc}
	\min\limits_{\pi \in \Pi} h(\lambda, \pi) = \min\limits_{\pi \in \Pi} \lim\sup\limits_{T\rightarrow \infty} \frac{1}{T} \mathbb{E}^\pi \left(\sum\limits_{t=0}^{T-1} C_{t} + \lambda\alpha_t c \right)\text{.}
\end{align} 
A policy that achieves $\mathcal{L}^*(\lambda)$ is called $\lambda$-optimal, denoted by $\pi_\lambda^{*}$, and is a solution to the following optimization problem
	$\min\limits_{\pi_\lambda} \mathcal{L}(\pi,\lambda)\text{.}$
Since the dimension of the state space $\mathcal{S}$ is finite, the \textit{growth condition} \cite[Eq. 11.21]{altman1999constrained} is satisfied. In addition, the immediate cost function is bounded below ($\geq 0$). Since these conditions are satisfied, the optimal value of the \ac{CMDP} problem in \eqref{optproblem},
$\bar{C}_\pi^{*}$, and the optimal value of the \eqref{opt:dual}, $\mathcal{L}^{*}(\lambda)$, ensure the following relation \cite[Corollary 12.2]{altman1999constrained}
\begin{align}
	\bar{C}_{\pi^*} = \sup\limits_{\lambda\geq 0} \mathcal{L}^{*}(\lambda)\text{.}
\end{align}

\begin{thm}[Mixture of two randomized policies]\label{thm:randomizedpolicies}\cite[Theorem 4.4]{beutler1985optimal}
 The optimal policy $\pi^{*}$ is a mixture of two deterministic policies $\pi^{*}_{\lambda^{-}}$, $\pi^{*}_{\lambda^{+}}$.
\end{thm}
The optimal policy is written symbolically as 
	$\pi^{*} = \eta \pi^{*}_{\lambda^{-}} + (1-\eta)\pi^{*}_{\lambda^{+}}\text{,}$
where $\eta$ is a probabilistic factor.
We characterize $\eta$, $\lambda^{-}$, and $\lambda^{+}$, later in this paper.  

We now proceed to find the solution to the problem \eqref{opt:dual}. To obtain the optimal policy of an infinite horizon average cost \ac{MDP}, it is sufficient to solve the following Bellman equation \cite{BertsakasDPII}
\begin{align}\nonumber\label{BellmanBerts}
	&\theta_\lambda  + V(S_t)  =  \\ 
	&\min\limits_{\alpha_t \in \mathcal{A}}\left\{ C_{a_t}  + \lambda\alpha_t c+\sum\limits_{S_{t+1}\in\mathcal{S}} P_{S_t,S_{t+1}}  V(S_{t+1}) \right\}\text{,} 
\end{align}
where $\theta_\lambda$ is the optimal value of \eqref{opt:dualWithoutc}, for a given $\lambda>0$, and $V(S_{t+1})$ is the \textit{cost-to-go}  or \textit{value function}.
This is known to be a challenging problem \cite{BertsakasDPII}. We apply the value iteration algorithm and the bisection method to solve the problem and to find the optimal Lagrange multiplier, respectively. The detailed steps are provided in Algorithm $1$ \footnote{There is no closed form expression for $\eta$ \cite{ma1986estimation}, thus we numerically search for $\eta \in \left[0,1\right]$.}.
\begin{algorithm}\label{alg:cmdp}
	\nonumber
	\caption{Value Iteration Algorithm}
	Initialization: $\lambda\leftarrow 0$, $\lambda_{-} \leftarrow 0$, $\lambda_{+}$ large positive number, and $\epsilon>0$\\
	Run $\text{VI}(\lambda)$\\
	\eIf{$\bar{c}\leq c_\text{max}$}{
	$\pi^{*} \leftarrow \pi^{*}_\lambda$}
	{\While{$|\lambda_{+} - \lambda_{-}|>\epsilon$}{
		\text{Run} $\text{VI}(\frac{\lambda_{+}-\lambda_{-}}{2})$\\
		\eIf{$\bar{c}\geq c_\text{max}$}{
			$\lambda_{-} \leftarrow \lambda$}
			{$\lambda_+ \leftarrow \lambda$}	
	}
	 $\lambda^{*} \leftarrow \frac{\lambda_{+} + \lambda_{-}}{2}$, $\lambda^{*}_{+} \leftarrow \lambda_{-}$, $\lambda_{+}^{*}\leftarrow \lambda_{-}$\\
 	  $\text{VI}(\lambda^{*})$\\
 	  \eIf{$\bar{c}=c_\text{max}$}{$\pi^* = \pi^*_\lambda$}
   	   {$\pi^{*} = \eta\pi_{\lambda^{*}_{-}}   + (1-\eta)\pi_{\lambda^{*}_{+}} $   }
	}
	\algrule[1pt]
	\textbf{function} $\text{VI}(\lambda)$: 
	\\
	Initialization: $V^{0} = V^{1} = 0$, $\forall s \in \mathcal{S}$, choose a small $\epsilon>0$, set $n=1$\\
	\While{$||V^n - V^{n-1}||\geq \epsilon (1-\gamma)/2\gamma $}{
		\For{ each $s\in \mathcal{S}$ compute}{
		
		\scriptsize
		$V^{n}(s) = \min\limits_{\alpha_t} \left\{C_{\alpha_t} + \lambda\alpha_tc + \gamma\sum\limits_{S_{t+1}\in \mathcal{S}}  \Pr\left\{S_{t+1}|S_t,\alpha_t   \right\}V^{n-1}(S_{t+1})   \right\}$}
		$n\leftarrow n + 1$
	}
\textbf{return} policy $\pi$ 
\end{algorithm}

\subsection{Suboptimal low-complexity algorithm}
Although the value iteration algorithm is proven to converge to the optimal solution, it suffers from high computational complexity, known as the \emph{curse of dimensionality} \cite{powell2007approximate}. Our goal is to provide a low-complexity algorithm that guarantees that the average cost constraints are satisfied and which provides a solution close to the optimal one. Using tools from Lyapunov optimization theory, we provide a real-time algorithm named \ac{DPP}. We reformulate the problem in \eqref{optproblem}, and we define the objective function $g(t)$ as
\begin{align}\label{eq:objectivelyapunov}\nonumber
	&g(t) = \\&
	\begin{cases}
			\left( \sum\limits_{k=i}^{N} C_{k,j} P_{i,k}  (1-p_s) + \sum\limits_{k=i}^N C_{k,i}P_{i,k} p_s\right) \text{, if } \alpha_t = 1\text{,} \vspace{3mm}\\
			\left( \sum\limits_{k=i}^N C_{k,j} \mathcal{P}_{i,k} \right)\text{, otherwise.}
	\end{cases}
\end{align}
The expected time average of the objective function is defined as 
$
	\bar{g} \triangleq \lim \sup \limits_{T\rightarrow \infty} \frac{1}{T} \sum\limits_{t=1}^T \mathbb{E} \left\{g(t)\right\}\text{.}
$
The reformulated stochastic optimization problem is the following
\begin{equation}
	\label{optproblemLyapuno}
	\min\limits_{\alpha_t} \quad \bar{g}\text{, } \text{ }\text{ s.~t., }   \bar{c}\leq c_\text{max}\text{.} 
\end{equation}
In order to satisfy the average cost constraints, we map the average cost constraint in eqref{optproblemLyapuno} into a virtual queue \cite{NeelyBook}. We show below that the time average cost problem is transformed into a queue stability problem. 

Let $\left\{Z(t)\right\}$ be the virtual queue associated with constraint $\eqref{optproblemLyapuno}$. The virtual queue is updated at every time slot $t$ as
\begin{align}\label{virtualqevolution}
	Z(t+1) = \max[Z(t)-c,0] + \alpha_t c\text{.} 
\end{align}
Process $\left\{Z(t)\right\}$ can be viewed as a virtual queue with arrivals $\alpha_t$ and service rate $c$. This idea is based on the fundamental \textit{Lyapunov drift} theorem \cite{meyn2012markov}. 

%
With the above definitions in mind, we can now proceed to describe our proposed algorithm and provide performance guarantees regarding the average cost constraint.
\begin{lemma}\label{lem:ratestable} If $Z(t)$ is rate stable\footnote{A discrete time process $Q(t)$ is \textit{rate stable} if
				$\lim\limits_{t\rightarrow \infty} \frac{Q(t)}{t} = 0\text{, with probability 1.}$}, then the constraint in \eqref{optproblemLyapuno} is satisfied.
\end{lemma}
\begin{proof}
	By using the basic sample property \cite{NeelyBook}[Lemma 2.1, Chapter 2], we obtain:
	\begin{align}\label{ineq:sampleproperty}
		\frac{Z(t)}{t} - \frac{Z(0)}{t} \geq \frac{1}{t}\sum\limits_{\tau = 0}^{t-1} c \alpha_t - \frac{1}{t} \sum\limits_{\tau = 0}^{t-1} c\text{.}
	\end{align}
If $Z(t)$ is rate stable, then $\lim\limits_{t\rightarrow \infty} \frac{Z(t)}{t} = 0$. By taking the time average expectations in $\eqref{ineq:sampleproperty}$ on both sides, we obtain the result. 
\end{proof}
In order to stabilize virtual queue $Z(t)$, and therefore by Lemma \ref{lem:ratestable} to satisfy the average cost constraints, we first define the Lyapunov function as 
	$L (Z(t)) \triangleq \frac{1}{2} Z^2(t)$
and the Lyapunov drift as
\begin{align}\label{eq:drift}
	\Delta(Z(t)) \triangleq \mathbb{E} \left\{L(Z(t+1)) - L(Z(t)) | Z(t)\right\}\text{.}
\end{align}
The above conditional expectation is with respect to the random source state transitions, channel states, and transmission decisions. We apply the \ac{DPP} algorithm to minimize the time average expected cost while stabilizing the virtual queues, $Z(t)$. Specifically, this approach seeks to minimize an upper bound on the following expression
\begin{align}\label{exprDriftPenalty}
	\Delta (Z(t)) + W\mathbb{E}\left\{g(t)\right\}\text{,}
\end{align}
where $W>0$ is an importance factor to scale the penalty. By utilizing $(\max\left[Q-b,0\right] + A)^2 \leq Q^2 + A^2 +b^2 + 2Q(A-b)$, we get the following upper bound on the expression in \eqref{exprDriftPenalty}
\begin{align}\nonumber
	&\Delta(Z(t)) + W\mathbb{E}\left\{g(t)\right\} \\ 
	&\leq B + W\mathbb{E}\left\{g(t)\right\} + \mathbb{E}\left\{Z(t)(c\alpha_t - c_\text{max})\right\}\text{,}
\end{align}
where $B < \infty$, and $B\geq \frac{(\alpha_t c_\text{max})^2 + c_\text{max}^2  }{2}.$

\subsection{Drift-Plus Penalty Algorithm}
At every time slot $t$, the transmitter observes the state of the source $(X_t)$ and the estimated state at the destination $(\bar{X}_t)$, and it takes a decision according to the following optimization problem
	\begin{align}\label{driftminimization}
		\min\limits_{\alpha_t} \quad & Wg(t) + Z(t)(c\alpha_t - c_\text{max})\text{.}
	\end{align}
\begin{lemma}\label{lem:boundness}
We consider a class of stationary policies, possibly randomized, denoted by $\Omega$. A policy $\omega(t)$ that belongs to the class $\Omega$ is an i.i.d. process that takes probabilistic decisions independent of the state of the system, at every time slot $t$.
Let $y(t) = c\alpha_t - c_\text{max}$, and $c(t) = c\alpha_t$. Then, if the problem in \eqref{optproblemLyapuno} is strictly feasible, and the second moments of $y(t)$ and $g(t)$ are bounded, then there is $\epsilon > 0$ for which there is an $\omega(t)$ policy such that the following holds
	\begin{align}\nonumber
		\mathbb{E} \left\{y(t)\right\} \leq \epsilon\text{, }
		\mathbb{E} \left\{g^{*}(t) \right\}  = g_\omega \leq g^\text{opt} + \epsilon\text{,}
	\end{align}
where $y^{*}(t)$ and $g^{*}(t)$ are the resulting values of the $\omega$ policy, and $g^{\text{opt}}$ is the optimal value function in \eqref{optproblemLyapuno} achievable by any optimal stationary randomized policy. 
\end{lemma}
\begin{proof}
Since the cost for sampling and transmission is bounded, the second moment of $c(t)$ is also bounded. Furthermore, since the values of the matrix $C$ are bounded, the second moment of $g(t)$ is also bounded. Therefore, we have
	\begin{align}\nonumber
		\mathbb{E}\left\{c(t)^{2}\right\}\leq c^2\text{, }
		\mathbb{E}\left\{g(t)^2\right\}  \leq C_\text{max}^2\text{,}
	\end{align}
where $C_\text{max}$ is the maximum value of instantaneous actuation cost. Then, the boundedness assumptions in \cite{NeelyBook}[Ch. 4.2.1] are satisfied. Therefore, from Theorem 4.5 in \cite{NeelyBook}, we get the result. 
\end{proof}
\begin{thm}
The \ac{DPP} algorithm satisfies any feasible set of average cost constraints.
\end{thm}
\begin{proof}
Since the \ac{DPP} algorithm seeks to minimize the expression in \eqref{driftminimization}, we obtain that 
	\begin{align}
		&\Delta(Z(t)) + W\mathbb{E}\left\{g(t)|S_t\right\} \\
		&\leq B + Z(t)\mathbb{E}\left\{y_\text{DPP}(t)\right\} + W\mathbb{E}\left\{g_\text{DPP}(t)\right\} \\ 
		& \leq B + Z(t)\mathbb{E}\left\{ y^{*}(t) \right\} + W\mathbb{E}\left\{g^{*}(t)\right\}\text{,}
	\end{align}
where $y^{*}(t)$ and $g^{*}(t)$ are the resulting values after applying policy $\omega$. By considering the bound in Lemma \ref{lem:boundness}, we get
\begin{align}
	\Delta(Z(t))  + W\mathbb{E}\left\{g(t)|S_t\right\} \leq B + \epsilon Z(t) + W(g^{\text{opt}} + \epsilon) \text{,}
\end{align}
and taking $\epsilon \rightarrow 0$, we have 
\begin{align}
	\Delta(Z(t)) + W\mathbb{E}\left\{g(t)|S_t\right\} \leq B + Wg^{\text{opt}}\text{.}
\end{align} 
The above expression is in the exact form of the Lyapunov optimization theory \cite{NeelyBook}[Theorem 4.2]. Therefore, the virtual queue is mean rate stable, and the average constraints are satisfied.
\end{proof}

\section{Simulation Results}
In this section, we compare the performance of the two proposed algorithms with the baseline algorithm proposed in \cite{pappas2021goal} in terms of average real-time reconstruction error and cost of actuation error. The baseline policy decides on sampling and transmission whenever there is a discrepancy between the states at the source and at the destination, i.e., $X_t\neq \bar{X}_t$. In our setup, the average cost constraint, $c_\text{max}$, is set to $0.2$, with a cost of sampling and transmission equal to $1$. Therefore, a feasible policy decides $20\%$ of the time for sampling and transmission. Note that the baseline algorithm does not take into account the sampling and transmission costs.

We consider two cases for the source dynamics: a slowly varying source and a fast varying source. We consider that the Markov source has four states and is modeled as a birth-death process. The cost of the actuation error matrix remains fixed in both cases. The values of the matrix are shown below
\begin{align}
    C =\left(\begin{array}{l|cccc}
         & 0 & 1 & 2 & 3\\
        \hline
       0 & 0 & 10 & 50 & 30\\
       1 & 10 & 0 & 40 & 20\\
       2 & 20 & 10 & 0 & 10\\
       3 & 30 & 20 & 40 & 0
    \end{array}\right)\text{,}
\end{align}
where element $C_{i,j}$ is the cost of actuation error for the source being in state $i$ while the estimated value, $\hat{X}$, is $j$. 

\subsection{Slowly varying source}
In Fig. \ref{figs:slowvarying}, we compare the average reconstruction error and the average cost of actuation error resulting from the three algorithms for the case of a slowly varying source. The transition of the Markov source is shown below:
\begin{align}
    P = \left[ \begin{matrix}
        0.8 & 0.2 &  0  &  0\\
        0.1 & 0.8 & 0.1 &  0\\
         0  & 0.1 & 0.8 & 0.1\\
         0  &  0  & 0.2 & 0.8
    \end{matrix}\right].
\end{align}

\begin{figure}%
	\centering
	\subfloat[\centering]{{\includegraphics[width=3.9cm]{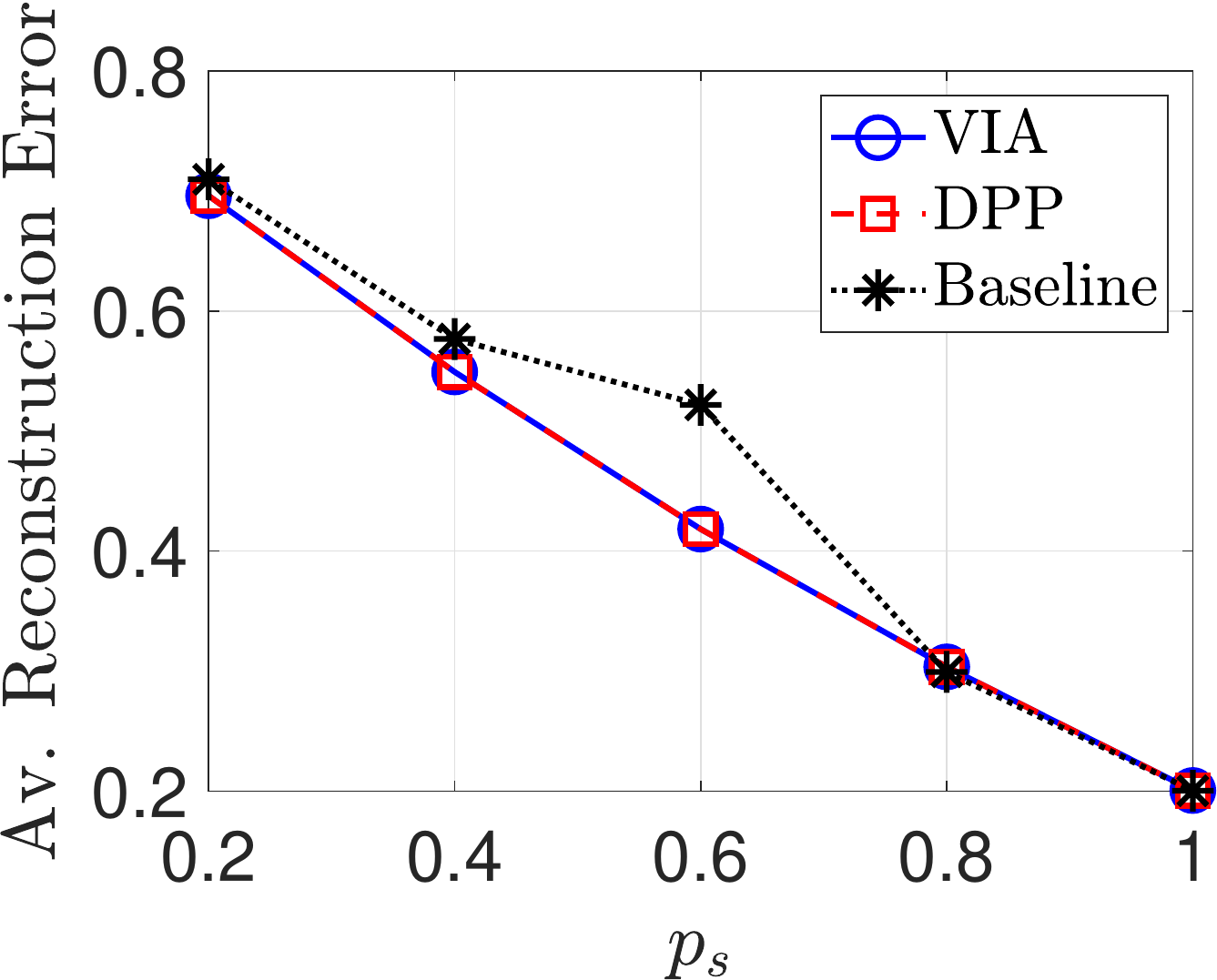} }} \label{fig:recerrorslow}%
	\qquad
	\subfloat[\centering]{{\includegraphics[width=3.9cm]{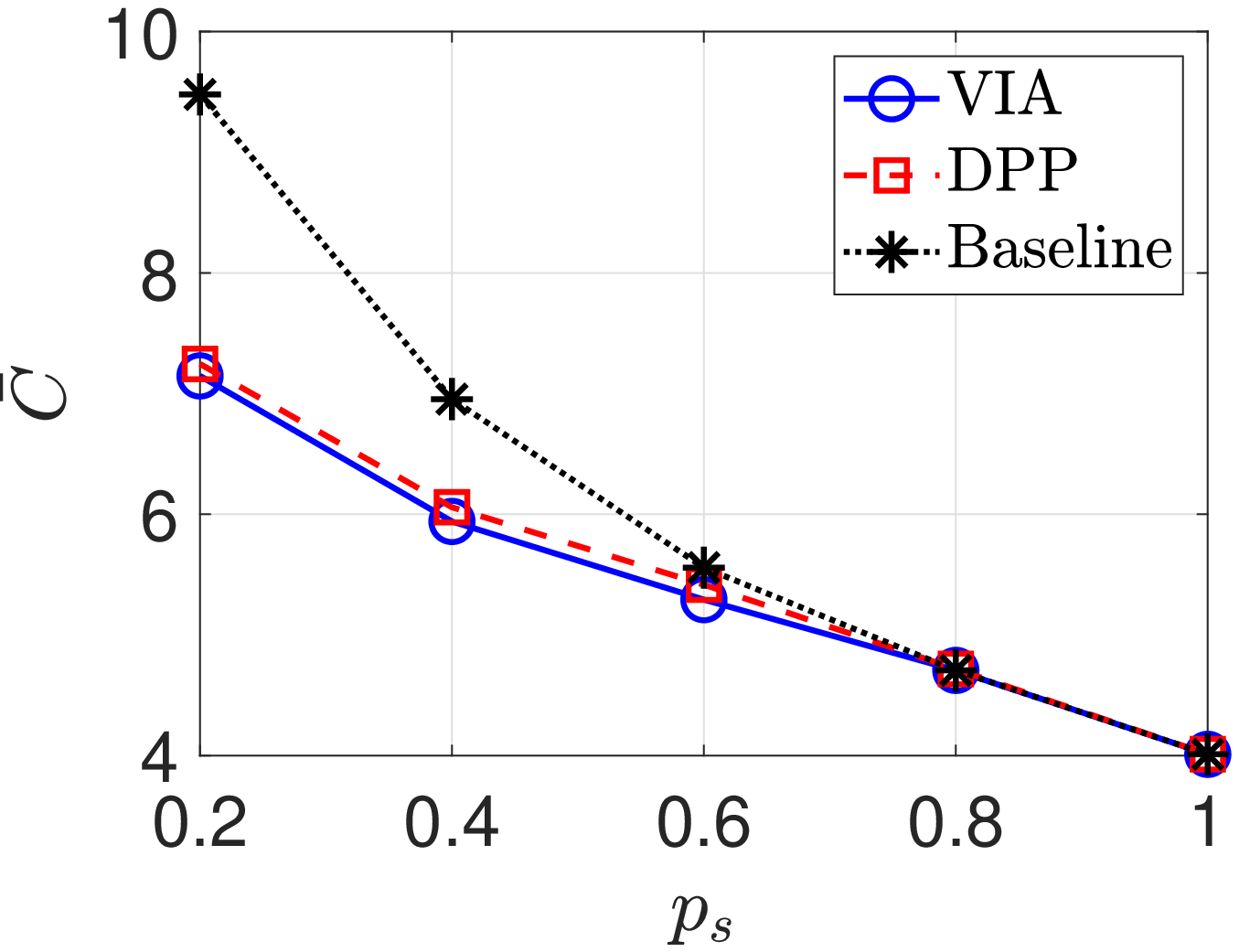} }}\label{fig:costactslow}%
	\caption{Performance of proposed policies as a function of the success probability for a slowly varying source.}%
	\label{figs:slowvarying}
\end{figure}

We consider that the source remains in the same state with a high probability $(0.8)$. In Fig. \ref{figs:slowvarying}a, we observe that the proposed \ac{DPP} and VIA algorithms have very similar performance and provide lower average reconstruction error than the baseline algorithm. Moreover, in Fig. \ref{figs:slowvarying}b, we see that the difference between the cost of actuation error performance of the baseline and the proposed algorithms increases. This is because the baseline algorithm decides on sampling and transmission whenever there is a discrepancy between the source and the destination, without explicitly taking into account the cost of actuation error.

\subsection{Rapidly varying source}
In the case of a rapidly varying source, we consider that the Markov source remains in the same state with probability $0.2$. The corresponding transition matrix is 
\begin{align}
    P = \left[ \begin{matrix}
        0.2 & 0.8 &  0  &  0\\
        0.4 & 0.2 & 0.4 &  0\\
         0  & 0.4 & 0.2 & 0.4\\
         0  &  0  & 0.8 & 0.2
    \end{matrix}\right].
\end{align}
In Figs. \ref{figs:fastvarying}a and \ref{figs:fastvarying}b, we provide results for the average reconstruction error and average cost of actuation error, respectively. We observe that the performance of the baseline algorithm is better than that of VIA and DPP as far as the reconstruction error is concerned. However, the proposed algorithms still provide superior performance in terms of cost of actuation error, which is the metric of interest in this paper. The reason is that the algorithms proposed here take into account both the cost of actuation error and the statistics of the Markov source. Therefore, a main observation from the results is that \emph{low reconstruction error does not necessarily imply a low average cost of actuation error}. The reason is that these are two different performance metrics, and in a remote monitoring system with delayed measurements, it is crucial to take into account the statistics of the source and how the states of the source vary with the time horizon. 

\begin{figure}%
	\centering
	\subfloat[\centering]{{\includegraphics[width=3.9cm]{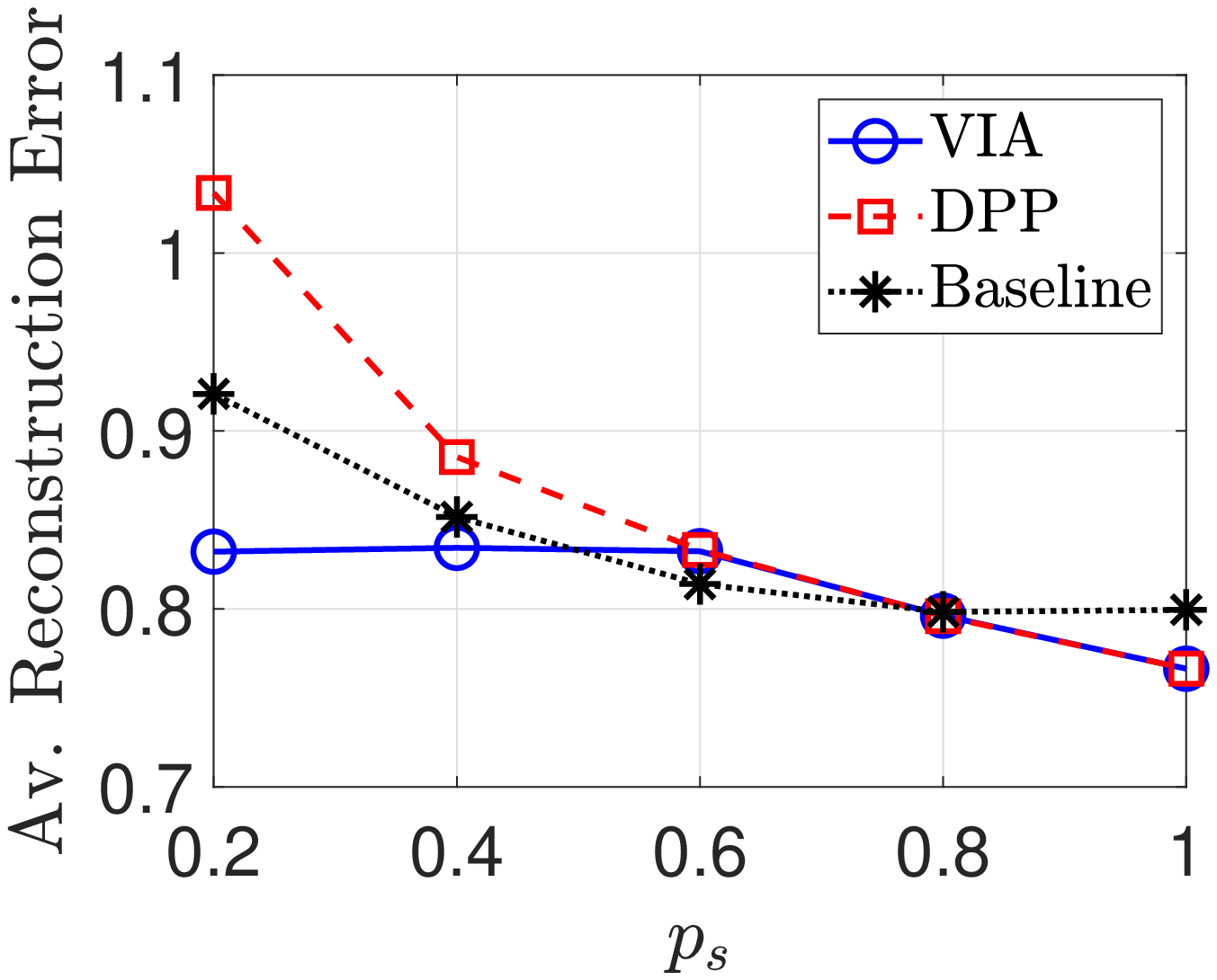}}} \label{fig:recerrorfast}%
	\qquad
	\subfloat[\centering]{{\includegraphics[width=3.9cm]{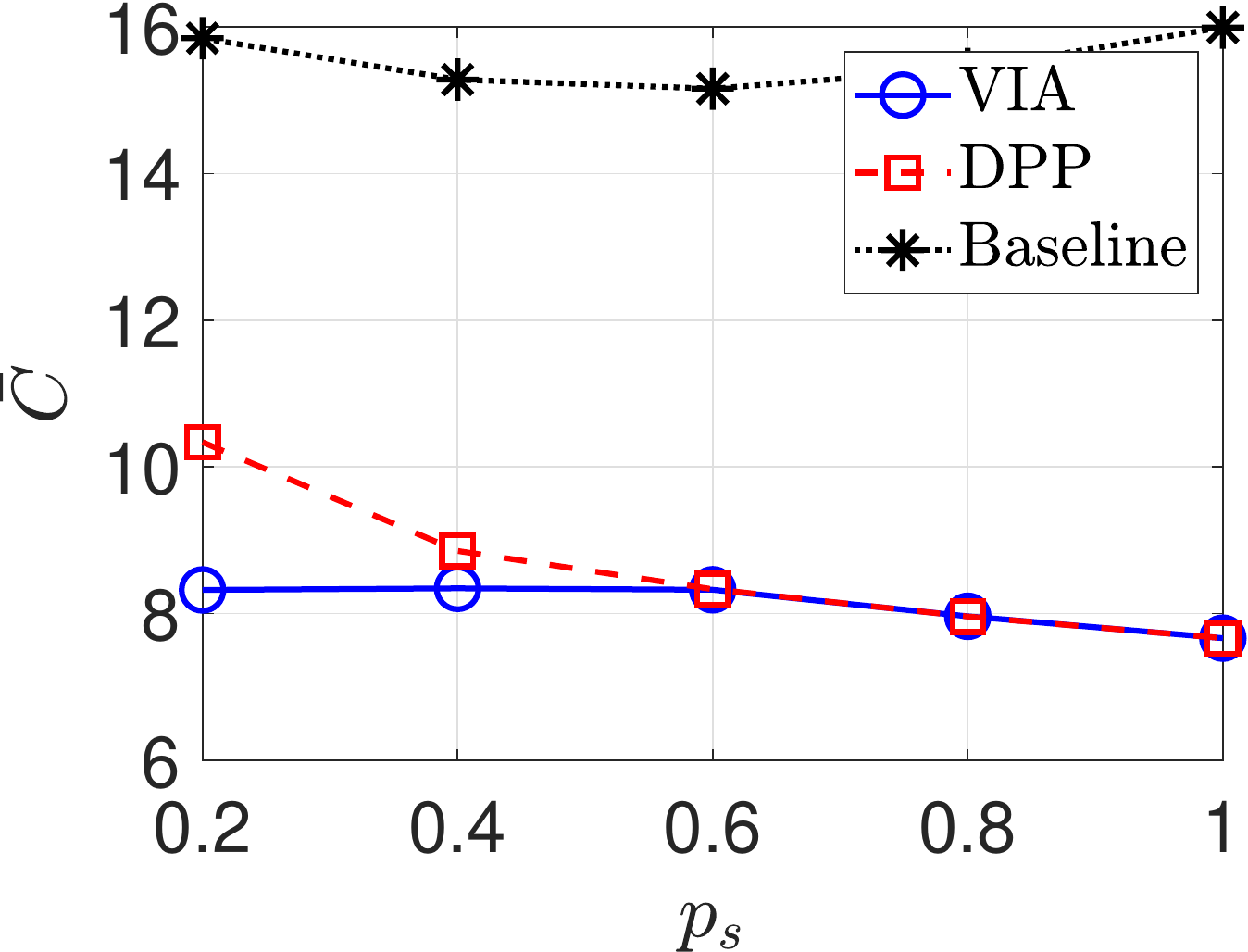} }}\label{fig:costactsfast}%
    \caption{Performance of proposed policies as a function of the success probability for a rapidly varying source.}
    \label{figs:fastvarying}
\end{figure}

\section{Conclusion}
In this work, we studied the minimization of the actuation error under resource constraints for real-time tracking of a remote source over wireless. We provided an optimal solution to the optimization problem and a low-complexity algorithm that guarantees the satisfaction of the average cost constraints. Our simulation results show that the performance of the low-complexity algorithm is close to optimal. We observed that an optimal policy for this problem takes into account not only the discrepancy in state between the source and the destination, but also the cost of actuation error that occurs due to this discrepancy, as well as the statistics of the source. Depending on the setup, it is sometimes beneficial to remain silent rather than perform sampling and transmitting a status update to the destination when there are delayed measurements sent to the actuator.


\bibliography{MyBib}
\bibliographystyle{ieeetr}
		
\end{document}